\documentclass[journal]{IEEEtran}
\usepackage{multicol}
\usepackage{graphicx}
\usepackage[latin1]{inputenc}
\usepackage{amsmath}
\usepackage{amsfonts}
\usepackage{acronym}
\usepackage{amsthm}
\usepackage{amssymb}
\usepackage{dsfont}    
\usepackage{epsfig}
\usepackage{ifthen}
\usepackage{psfrag}
\usepackage{url}
\usepackage{listings}
\usepackage{comment}
\usepackage{fancyhdr}
\usepackage[newitem,newenum]{paralist}
\usepackage{booktabs}
\usepackage{tabularx}
\usepackage{color}
\usepackage{cite}
\usepackage{morefloats}
\usepackage{setspace}
\usepackage[newitem,newenum]{paralist}
\usepackage{booktabs}
\usepackage[normalsize]{subfigure}
\newtheorem{theorem}{Theorem}
\newtheorem{proposition}{Proposition}

\usepackage{color}
\newcounter{Mytempeqncnt}

\begin{document}
%
\title{Relay Placement for Physical Layer Security: A Secure Connection Perspective}
\author{Jianhua~Mo, Meixia~Tao,~\IEEEmembership{Senior~Member,~IEEE}, and Yuan~Liu,~\IEEEmembership{Student~Member,~IEEE}
\thanks{Manuscript received March 16, 2012; revised April 5, 2012. The associate editor coordinating the
review of this letter and approving it for publication was Kai Kit Wong.}
\thanks{The authors are with the Dept. of Electronic Engineering, Shanghai Jiao Tong University, P. R. China(email:\{mjh, mxtao, yuanliu\}@sjtu.edu.cn).}
\thanks{This work is supported by the Innovation Program of Shanghai Municipal Education Commission under grant 11ZZ19 and the Joint Research Fund for Overseas Chinese, Hong Kong and Macao Young Scholars under grant 61028001.}
}
\maketitle

\begin{abstract}
This work studies the problem of secure connection in cooperative wireless communication with two relay strategies, decode-and-forward (DF) and randomize-and-forward (RF).
The four-node scenario and cellular scenario are considered.
For the typical four-node (source, destination, relay, and eavesdropper) scenario,
we derive the optimal power allocation for the DF strategy and find that the RF strategy is always better than the DF to enhance secure connection.
In cellular networks, we show that without relay, it is difficult to establish secure connections from the base station to the cell edge users. The effect of relay placement for the cell edge users is demonstrated by simulation.
For both scenarios, we find that the benefit of relay transmission increases when path loss becomes severer.

\end{abstract}


\setlength\arraycolsep{2pt}

\section{Introduction}
Wireless communication is inherently vulnerable to eavesdropping due to its broadcast nature. However, by exploiting the randomness of the wireless propagation channels, we can enhance the security in physical layer\cite{Wyner_Bell75}.
On the other hand, cooperative relay has received much attentions due to its ability of power reduction, coverage extension, and throughput enhancement. Thus, it is attractive and promising to utilize these benefits for physical layer security.

The authors in \cite{Lai_IT08} discussed the four-node (source, destination, relay, eavesdropper) secure communication system from an information-theoretical perspective and studied several relay strategies, such as decode-and-forward (DF) and  noise-forwarding (NF).
Authors in \cite{Jeong_TSP11} investigated the secrecy rate maximization problem for the four-node system in multicarrier relay channel with the DF strategy.
For the secure transmission system with multiple relays, the beamforming and relay selection was considered in \cite{Dong_TSP10} and \cite{Krikidis_TWC09} respectively under the assumption that the eavesdropper only wiretaps the second hop during the cooperative transmission.
A joint problem of secure resource allocation and scheduling was studied in \cite{Ng_TWC11} for cellular networks with DF relays.

In \cite{Koyluoglu_IT12}, the authors proposed another relay strategy in which the relays add independent randomization in each hop (we refer it as randomize-and-forward (RF)).  It was proved therein that under the RF strategy, securing each individual hop is sufficient for securing the end-to-end transmission. Scaling law of secrecy capacity were then obtained by using such RF strategy in \cite{Koyluoglu_IT12}. The authors in \cite{Goeckel_JSAC11} analyzed the maximal number of eavesdroppers that can be tolerated in the two-hop secure transmission with jamming when RF strategy was used.

Our paper is motivated twofold. First, though the cooperative secure transmission has been studied in several scenarios (e.g., \cite{Lai_IT08, Jeong_TSP11,Krikidis_TWC09, Dong_TSP10, Ng_TWC11, Koyluoglu_IT12, Goeckel_JSAC11, Popovski_TIFS09}), to our best knowledge, no attempt has been made to study relay placement for physical layer security.
Second, although fading was utilized to achieve physical layer security (e.g., \cite{Bloch_IT08}), there is no theoretical analysis about the impacts of large scale path loss on security.

The main contributions of this work are summarized as follows. 1) In the four-node system, we derive the optimal power allocation for the DF strategy and find that the RF strategy is always better than the DF in terms of secure connection probability. 2) We show that when the eavesdropper is far away, placing the relay at the midpoint of the source and the destination is asymptotical optimal, and the outage probability of the RF strategy is about half of the DF. 3) In cellular networks, we derive the secure outage probability without relay and show the superiority of placing RF relay over DF relay through simulation. 4) We analyze the effects of path loss on secure connection  and find that relay transmission achieves more benefit when path loss is severer.

\begin{figure*}[t]
\normalsize
\setcounter{Mytempeqncnt}{ \value{equation}}
\setcounter{equation}{5}
\begin{equation}
\label{eqn_plague_in_lambda}
P_{DF}\left(\boldsymbol{d}\right) =1-\frac{{d_{se}^\alpha d_{re}^\alpha }}{{\left( {d_{rd}^\alpha  + d_{re}^\alpha } \right)\left( {d_{sr}^\alpha  + d_{se}^\alpha } \right) + d_{sr}^\alpha d_{rd}^\alpha  + \frac{{{p_r}d_{sr}^\alpha }}{{{p_s}}}\left( {d_{sr}^\alpha  + d_{se}^\alpha } \right) + \frac{{{p_s}d_{rd}^\alpha }}{{{p_r}}}\left( {d_{rd}^\alpha  + d_{re}^\alpha } \right)}}.
\end{equation}
\vspace*{4pt} \hrulefill
\setcounter{equation}{ \value{Mytempeqncnt}}
\end{figure*}

\section{Main Results}
\label{four_node_systems}
\begin{figure}[t]
\begin{centering}
\includegraphics[scale=.40]{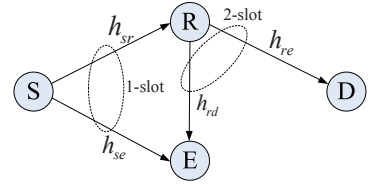}
\vspace{-0.1cm}
\centering
 \caption{An illustration of the four-node system model, where $S$, $R$, $D$, and $E$ represent the source, relay, destination, and eavesdropper, respectively.}\label{fig:Four_node}
\end{centering}
\vspace{-0.3cm}
\end{figure}
We consider two scenarios, i.e., the four-node system and cellular networks.
For both scenarios, we assume that the cooperative transmission consists of two phases. During the first phase,  the source (or base station (BS)) transmits while the relay (or relay station (RS)) listens. During the second phase, the relay transmits and the destination (or mobile user (MU)) listens. The eavesdropper overhears in both phases. Here we assume that the direct link from the source (or BS) to the destination (or MU) is not available. The wireless fading channels are modeled by large-scale fading with path loss exponent $\alpha$ and small-scale block Rayleigh fading.

\textsl{Notations}: Subscripts $s$, $r$, $d$ and $e$ represent the source (or BS), relay (or RS), destination (or MU) and eavesdropper, respectively. $d_{ij}$ and $h_{ij}$ denote the distance and channel coefficient between node $i$ and $j$, respectively. $p_s$ and $p_r$ denote transmit powers of the source and relay, respectively. For brevity, we denote $\boldsymbol{d}:= \left\{d_{sr}, d_{rd}, d_{se}, d_{re}\right\} $ and $\boldsymbol{p}:=\left\{ p_s, p_r \right\}$.
\subsection{Four-node System}

In this subsection, we study a four-node system consisting of a source, a destination,
an eavesdropper and a relay shown in Fig. \ref{fig:Four_node}.
Both DF and RF strategies are analyzed in terms of secure connection probability.
Here the knowledge of channel state information (CSI) for the eavesdropper is assumed to be known as the eavesdropper may be another legitimate user who transmits signals but is not allowed to receive the confidential message from the source \cite{Bloch_IT08}.
\subsubsection{Decode-and-Forward (DF)}
For the DF strategy, the relay uses the same codebook as the source's.
The achievable rate from the source to the destination is given by
\begin{equation}
R_{d} = \frac{1}{2} \min \left\{ \log_2 \left(1+\frac{{p_s {\left| {{h_{sr}}} \right|}^2 }}{{d_{sr}^\alpha }}\right),\log_2 \left(1+\frac{p_r {\left| {{h_{rd}}} \right|}^2}{d_{rd}^\alpha } \right)\right\}.
\end{equation}
The eavesdropper wiretaps and combines the signals from both two hops, and as such the information rate at the eavesdropper is
\begin{equation}
R_{e} = \frac{1}{2} \log_2 \left(1+\frac{p_s {\left| {{h_{se}}} \right|}^2}{{d_{se}^\alpha }} + \frac{p_r {\left| {{h_{re}}} \right|}^2 }{d_{re}^\alpha }\right).
\end{equation}
The secrecy rate of the system is
\begin{equation}
R_s = \max \left\{R_{d}-R_{e}, 0\right\}.
\end{equation}

Similar to \cite{Pinto_TIFS12a,Zhou_TWC11}, we define that the connection between the source and destination is secure if $R_s > 0$.
Then the secrecy outage probability can be defined as
\begin{eqnarray*}
& &P_{DF}\left( \boldsymbol{d}, \boldsymbol{p} \right) = Pr \left(R_s<0\right)  \label{eqn_P_out_DF_formulation} \\
&=& Pr \left( \min \left\{{\frac{p_s {|h_{sr}|}^2}{d_{sr}^\alpha },\frac{ p_r {|h_{rd}|}^2 }{ d_{re}^\alpha }}\right\} < \frac{p_s {\left| {{h_{se}}} \right|}^2}{{d_{se}^\alpha }} + \frac{ p_r {\left| {{h_{re}}} \right|}^2 }{ d_{re}^\alpha }\right). \nonumber
\end{eqnarray*}
\begin{proposition}
\label{proposition_DF}
For the DF strategy, the optimal power allocation satisfies
\begin{equation}
\label{eqn_DF_relay_power_allocation}
\frac{{p_r}}{{p_s}} = \sqrt{ \frac{{d_{rd}^\alpha \left( {d_{rd}^\alpha  + d_{re}^\alpha } \right)}}{{d_{sr}^\alpha \left( {d_{sr}^\alpha  + d_{se}^\alpha } \right)}}},
\end{equation}
and the minimal outage probability, denoted as $P_{DF}\left(\boldsymbol{d}\right)$, is
\begin{equation}
\label{eqn_DF_minimal_outage_probability}
P_{DF}\left( \boldsymbol{d} \right) =  1- \frac{{d_{se}^\alpha d_{re}^\alpha }}{{\left( \sqrt{\left( {d_{sr}^\alpha  + d_{se}^\alpha } \right) \left( {d_{rd}^\alpha  + d_{re}^\alpha } \right) }  + \sqrt{d_{sr}^\alpha d_{rd}^\alpha }  \right)}^2}.
\end{equation}
\end{proposition}
\begin{IEEEproof}
Note that $\frac{{{p_s {|h_{sr}|}^2}}}{{d_{sr}^\alpha }}, \frac{{{p_r {|h_{rd}|}^2}}}{{d_{rd}^\alpha }}, \frac{{{p_s {|h_{se}|}^2}}}{{d_{se}^\alpha }}$ and
$\frac{{{p_r {|h_{re}|}^2}}}{{d_{re}^\alpha }}$ are exponential distributed with means $\frac{{p_s}}{{d_{sr}^\alpha }}, \frac{{{p_r}}}{{d_{rd}^\alpha }}, \frac{{{p_s}}}{{d_{se}^\alpha }}$ and $\frac{{{p_r}}}{{d_{re}^\alpha }}$, respectively. Through some derivation, the outage probability is \eqref{eqn_plague_in_lambda} on the top of this page. Using the inequality of arithmetic and geometric means, we get \eqref{eqn_DF_relay_power_allocation} and \eqref{eqn_DF_minimal_outage_probability}.
\end{IEEEproof}
Proposition \ref{proposition_DF} shows that, to minimize the outage probability, only the power ratio $\frac{p_r}{p_s}$ matters rather than the absolute power.
\begin{figure}[t]
\begin{centering}
\makeatletter\def\@captype{figure}\makeatother
\subfigure[$P_{DF}(\boldsymbol{d})$]{\includegraphics[width=1.65in]{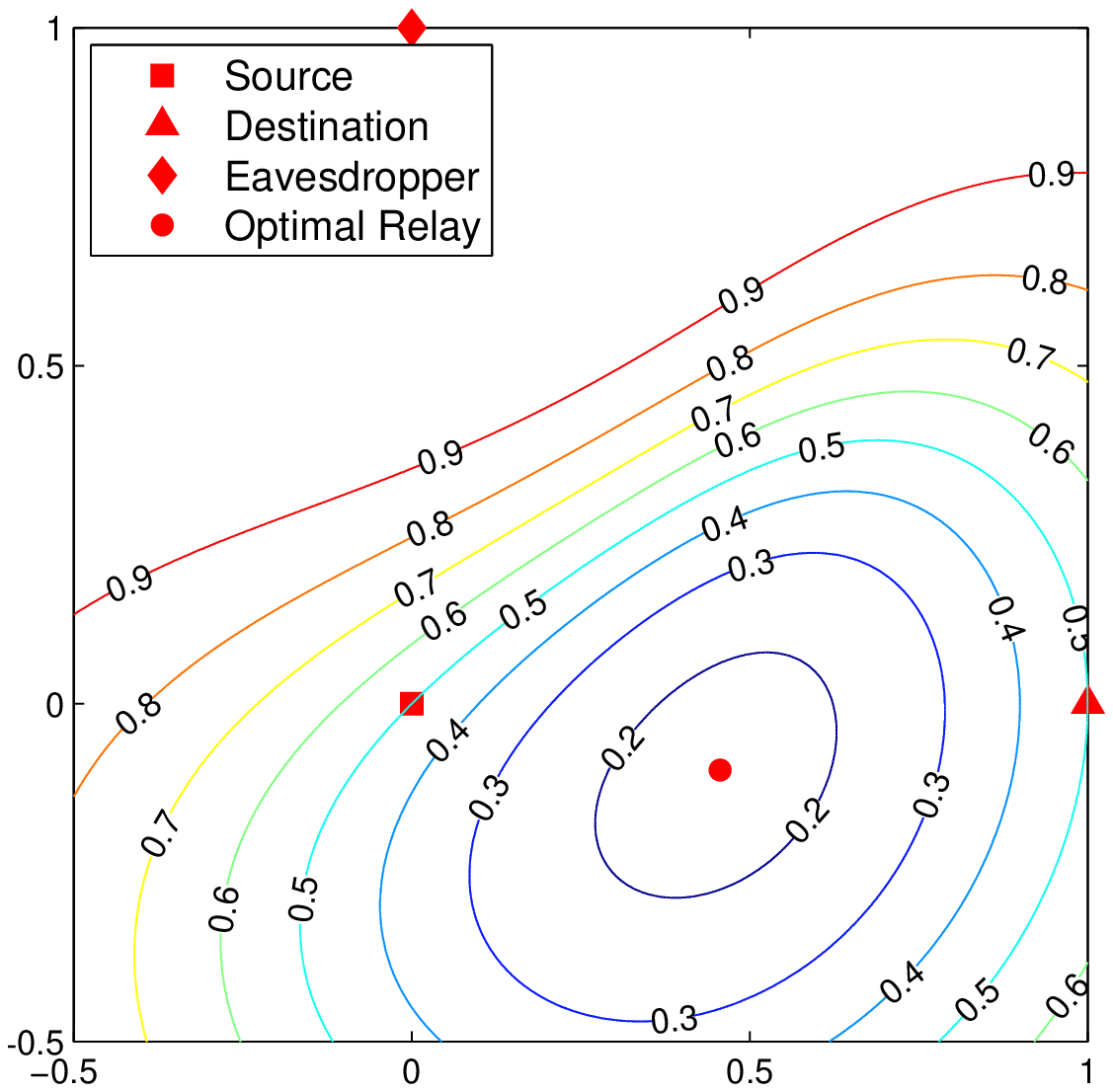} \label{fig_sim}}
\subfigure[$P_{RF}(\boldsymbol{d})$]{\includegraphics[width=1.65in]{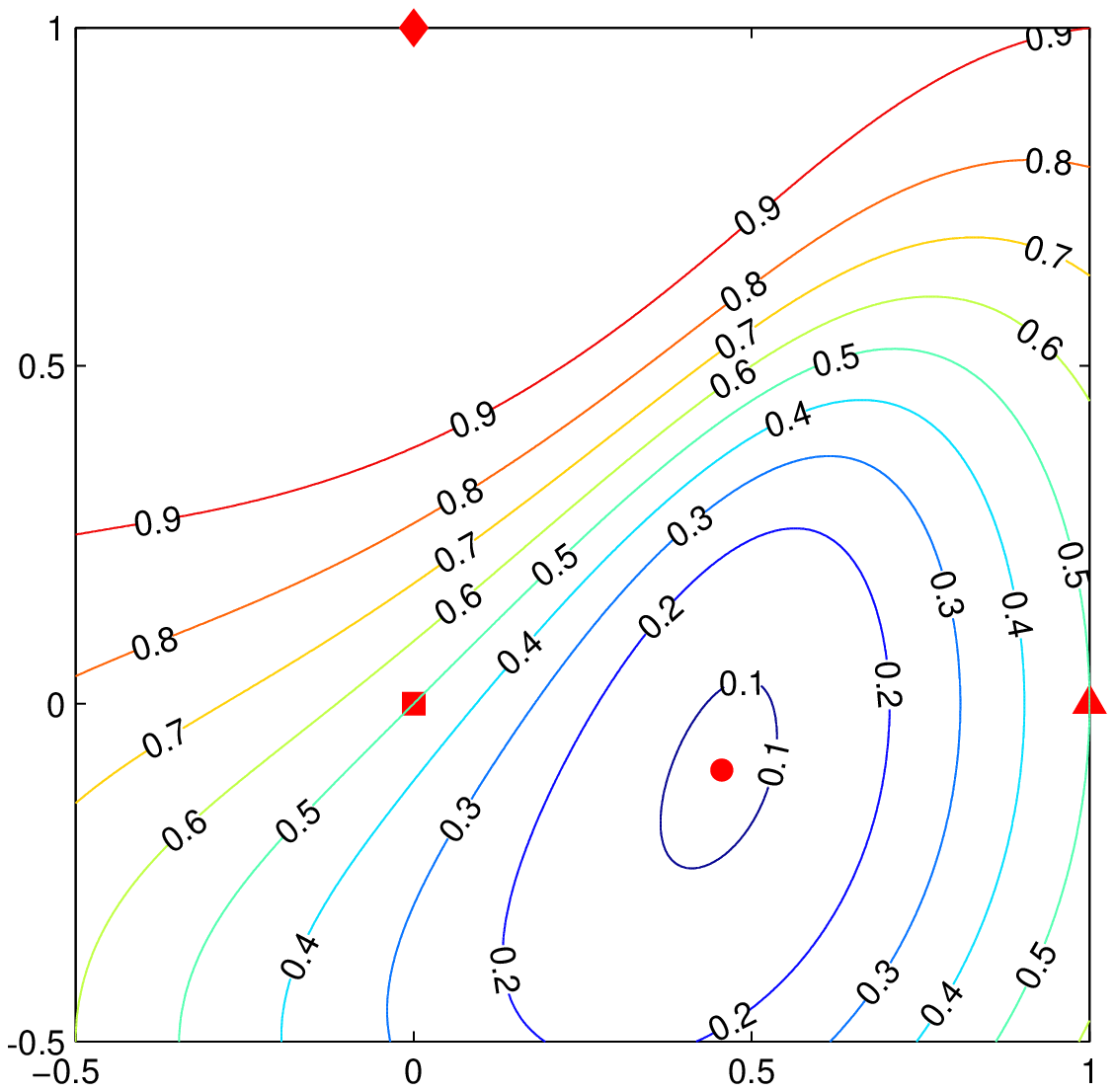} \label{fig_sim}}
\vspace{-0.1cm}
\caption{The outage probability as a function of the position of DF and RF relay. The source, destination and eavesdropper are at $(0,0)$, $(1,0)$ and $(0,1)$, respectively and $\alpha=4$. The optimal DF and RF relay positions are both around $(0.4551,-0.0987)$ with minimal $P_{DF}\left(\boldsymbol{d}\right)\approx 0.1645$ and $P_{RF}\left(\boldsymbol{d}\right)\approx 0.0878$ while $P_{Direct}\left(\boldsymbol{d}\right)=0.5$.}
\label{fig:DF_RF_relay_position_four_node}
\end{centering}
\vspace{-0.3cm}
\end{figure}
\subsubsection{Randomize-and-Forward (RF)}
For the RF strategy, the source and relay use different codebooks to transmit the secret message.
According to \cite{Koyluoglu_IT12}, the message is secured if the two hops are both secured.
Thus the outage probability can be defined as

\setcounter{equation}{6}

\begin{eqnarray}
\label{eqn_RF_outage_formulation}
P_{RF} \left(\boldsymbol{d}\right)&=& 1- Pr \left({\frac{{|h_{sr}|}^2}{d_{sr}^{\alpha}}>\frac{{|h_{se}|}^2}{d_{se}^{\alpha}}}\right) Pr \left(\frac{{|h_{rd}|}^2}{d_{rd}^{\alpha}}> \frac{{|h_{re}|}^2}{d_{re}^{\alpha}}\right) \nonumber \\
&=& 1 - \frac{d_{se}^{\alpha} d_{re}^{\alpha}}{ \left(d_{sr}^{\alpha}+ d_{se}^{\alpha}\right) \left(d_{rd}^{\alpha}+d_{re}^{\alpha}\right)}.
\end{eqnarray}

 \eqref{eqn_RF_outage_formulation} shows that for the RF strategy, the source and relay powers do not influence the outage probability, which is different from DF.

Since neither \eqref{eqn_DF_minimal_outage_probability} nor \eqref{eqn_RF_outage_formulation} is a convex function of the relay position, we resort to numerical results. In Fig. \ref{fig:DF_RF_relay_position_four_node}, we plot $P_{DF}\left(\boldsymbol{d}\right)$ and $P_{RF}\left(\boldsymbol{d}\right)$ as functions of the relay position. We find that the optimal positions of the DF and RF relays are both near to the midpoint of the source and destination. Moreover, the RF strategy is better than the DF strategy.
\begin{theorem}
\label{proposition_RF_better_DF}
  For the four-node system, the outage probability of the DF strategy is always larger than  that of RF strategy.
\end{theorem}
\begin{proof}
  Observing \eqref{eqn_DF_minimal_outage_probability} and \eqref{eqn_RF_outage_formulation}, we have

\begin{equation}
\label{eqn_relation_DF_RF}
    \sqrt { \frac{1}{1-P_{DF}\left(\boldsymbol{d}\right)} } = \sqrt { \frac{1}{1-P_{RF}\left(\boldsymbol{d}\right)} } + \sqrt {\frac{d_{sr}^{\alpha} d_{rd}^{\alpha}} {d_{se}^{\alpha} d_{re}^{\alpha}}}.
\end{equation}
Thus, $P_{DF}\left(\boldsymbol{d}\right) > P_{RF}\left(\boldsymbol{d}\right)$ and Theorem \ref{proposition_RF_better_DF} is proved.
\end{proof}
\begin{proposition}
\label{proposition_outage_large_guard_zone}
 When the eavesdropper is far away from the source and destination,\footnote{This scenario is applicable when the eavesdropper can not come closer to the legitimate nodes than a specified distance or when each legitimate node is able to physically inspect its surroundings and deactivate the nearby eavesdroppers \cite{Pinto_TIFS12a}.}
 the asymptotic optimal relay position is at the midpoint of the source and destination and
\begin{equation}
\label{eqn_outage_large_guard_zone_proposition}
  P_{RF}\left(\boldsymbol{d}\right) \approx \frac{1}{2} P_{DF}\left(\boldsymbol{d}\right) \approx \frac{1}{2^{\alpha-1}} P_{Direct}\left(\boldsymbol{d}\right),
\end{equation}
where $P_{Direct}\left(\boldsymbol{d}\right) = \frac{d_{sd}^{\alpha}}{d_{sd}^{\alpha} + d_{se}^{\alpha}} \approx \frac{d_{sd}^{\alpha}}{d_{se}^{\alpha}} $ is the outage probability of direct transmission.
\end{proposition}
\begin{IEEEproof}
   By \eqref{eqn_DF_minimal_outage_probability} and \eqref{eqn_RF_outage_formulation}, if $d_{se}\gg d_{sr}$ and $d_{re}\gg d_{rd}$,
  \begin{eqnarray}
    & &P_{DF}\left(\boldsymbol{d}\right) \approx \left( {\sqrt {\frac{{d_{sr}^\alpha }}{{d_{se}^\alpha }}}  + \sqrt {\frac{{d_{rd}^\alpha }}{{d_{re}^\alpha }}} } \right)^2 \approx \left( {\sqrt {\frac{{d_{sr}^\alpha }}{{d_{se}^\alpha }}}  + \sqrt {\frac{{d_{rd}^\alpha }}{{d_{se}^\alpha }}} } \right)^2, \label{eqn_DF_outage_large_guard_zone} \\
    & &P_{RF}\left(\boldsymbol{d}\right) \approx \frac{d_{sr}^{\alpha}}{d_{se}^{\alpha}} + \frac{d_{rd}^{\alpha}}{d_{re}^{\alpha}} \approx \frac{d_{sr}^{\alpha}}{d_{se}^{\alpha}} + \frac{d_{rd}^{\alpha}}{d_{se}^{\alpha}}. \label{eqn_RF_outage_large_guard_zone}
  \end{eqnarray}
  To minimize \eqref{eqn_DF_outage_large_guard_zone} and \eqref{eqn_RF_outage_large_guard_zone}, we should have
  \begin{equation}
    {d_{sr}}={d_{rd}}=\frac{1}{2} d_{sd}.
  \end{equation}
  Thus, the Proposition \ref{proposition_outage_large_guard_zone} follows.
\end{IEEEproof}
\begin{figure}[t]
\begin{centering}
\includegraphics[scale=.46]{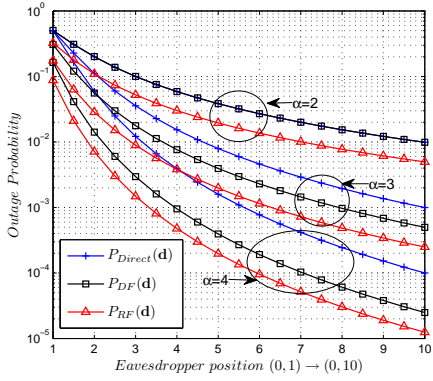}
\vspace{-0.1cm}
 \caption{Outage probability vs the eavesdropper's position. The source and destination are at $(0,0)$ and $(1,0)$, respectively.
 For every given eavesdropper's position, we find the optimal relay position and the corresponding minimal outage probability by numerical search.}\label{fig:Asymptotic_analysis}
\end{centering}
\vspace{-0.3cm}
\end{figure}
Fig. \ref{fig:Asymptotic_analysis} shows the outage probabilities when the eavesdropper moves away from the source and destination. In this figure, the asymptotic results of Proposition \ref{proposition_outage_large_guard_zone} are verified. It is first observed that the outage probability of the RF strategy is indeed about half of the DF. In addition, as the path loss exponent $\alpha$ increases, more benefit can be achieved from relay transmission. Notice that $P_{DF}\left(\boldsymbol{d}\right) \approx P_{Direct}\left(\boldsymbol{d}\right)$ if $\alpha=2$, meaning that DF relay transmission, compared with direct transmission,  brings no benefit at this time!
%

\subsection{Cellular Networks}
\label{cellular_networks}
We now consider a single-cell cellular network shown in Fig. \ref{fig:cellular_networks}. The hexagonal microcell is approximated as a circular cell of radius $R$ with a BS at the center of the cell. The MUs aim to get a secure connection with the BS. Only downlink is considered and uplink transmission can be encrypted by the key transmitted through the secure downlink. The eavesdroppers, which may be other MUs, do not cooperate and are uniformly distributed within the cell.
Therefore, the knowledge of CSI for the eavesdroppers are not needed.
\begin{proposition}
\label{proposition_cellular_outage_no_relay}
If there exist $N$ non-cooperative eavesdroppers uniformly distributed in the cellular networks, the outage probability for direct transmission, $P_{Direct}^{N}$, satisfies
\begin{equation}
\label{eqn_outage_N_Eve_1_Eve}
  P_{Direct}^{1}\left(d_{sd}\right) \leqslant P_{Direct}^{N}\left(d_{sd}\right) \leqslant 1 - \left(1- P_{Direct}^{1}\left(d_{sd} \right)\right)^N ,
\end{equation}
where $P_{Direct}^{1}\left(d_{sd}\right)$ is the outage probability with single eavesdropper and given by

\begin{eqnarray}
\label{eqn_outage_probability_direct_1}
 & &P_{Direct}^{1}\left(d_{sd} \right) = x^2 \sum\limits_{k = 0}^\infty  {\frac{{{{\left( { - 1} \right)}^k}}}{{1 + \frac{{k\alpha }}{2}}}} {\left( {\frac{1}{x^2}} \right)^{1 + \frac{{k\alpha }}{2} }} \label{eqn_cellular_direct_outage_probability}\\
 &=& \left\{ \begin{gathered}
  {x^2}\ln \left( {1 + \frac{1}{{{x^2}}}} \right) \quad when \quad \alpha = 2 \hfill \nonumber \\
  2{x^2}\left( {\frac{1}{6}\ln \frac{{\left( {{x^2} - x + 1} \right)}}{{{{\left( {x + 1} \right)}^2}}} + \frac{1}{{\sqrt 3 }}\left( {\arctan \frac{{2 - x}}{{\sqrt 3 x}} + \frac{\pi }{6}} \right)} \right) \nonumber \\
   \qquad when \quad \alpha = 3 \hfill \nonumber \\
  {x^2}\arctan \left( {\frac{1}{{{x^2}}}} \right) \quad when \quad \alpha = 4 \hfill \nonumber \\
\end{gathered}  \right.
\end{eqnarray}
with $x = {d_{sd}}/{R}$ is  the normalized distance between the BS and MU.
\end{proposition}

\begin{figure}[t]
\begin{centering}
\includegraphics[scale=.30]{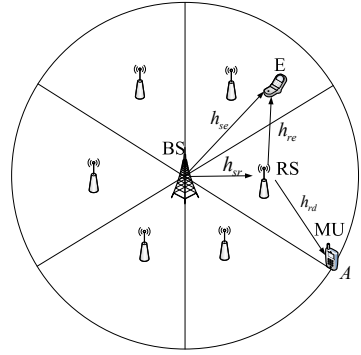}
\vspace{-0.1cm}
 \caption{An example of cellular networks with $6$ sectors and RS's. The RS's are placed on the angle bisector of each sector.}\label{fig:cellular_networks}
\end{centering}
\vspace{-0.3cm}
\end{figure}

\begin{IEEEproof}
The \textit{probability density function} of $d_{sd}$ is
\begin{equation}
  f\left(d_{sd}\right) = \frac{2d_{sd}}{R^2}, \quad 0\leqslant d_{sd} \leqslant R.
\end{equation}
We then can prove \eqref{eqn_outage_probability_direct_1}  by integrating
\begin{eqnarray*}
  P_{Direct}^{1}\left(d_{sd} \right) &=& \int_0^R {Pr\left(\frac{|h_{sd}|^2}{d_{sd}^{\alpha}} < \frac{|h_{se}|^2}{d_{se}^{\alpha}}\right) f(d_{se}) }d{d_{se}}.
\end{eqnarray*}
For \eqref{eqn_outage_N_Eve_1_Eve}, the left inequality is obvious. For the right one,
\begin{eqnarray*}
P_{Direct}^{N}(d_{sd}) &=& \displaystyle{\mathop{ \mathbb{E}}_{\{d_{se_{i}}, 1 \leqslant i\leqslant N \}}}\left\{Pr\left(\frac{|h_{sd}|^2}{d_{sd}^{\alpha}} < \max_{i} \frac{|h_{se}|^2}{d_{se_{i}}^{\alpha}}\right)\right\} \\
&=& \displaystyle{\mathop{ \mathbb{E} }_{\{d_{se_{i}}, 1 \leqslant i\leqslant N\}}}\left\{1 - Pr\left( \bigcap_{i=1}^{N}{\left(\frac{|h_{sd}|^2}{d_{sd}^{\alpha}} > \frac{|h_{se_{i}}|^2}{d_{se_{i}}^{\alpha}}\right)}\right)\right\} \\
&\stackrel{(a)}{\leqslant}& \displaystyle{\mathop{ \mathbb{E} }_{\{d_{se_{i}}, 1 \leqslant i\leqslant N\}}}\left\{1- \prod_{i=1}^{N}{Pr\left(\frac{|h_{sd}|^2}{d_{sd}^{\alpha}} >  \frac{|h_{se}|^2}{d_{se_{i}}^{\alpha}}\right)} \right\}  \\
&=& 1- \left(1- P_{Direct}^{1}\left(d_{sd} \right)\right)^N,
\end{eqnarray*}
where $e_{i}$ denotes the $i$th eavesdropper and $\mathbb{E}$ represents expectation. The inequality $(a)$ is obtained by using conditional probability.
This completes the proof of Proposition \ref{proposition_cellular_outage_no_relay}.
\end{IEEEproof}
$P_{Direct}^{N}$ can be obtained from numerical simulation. We plot $P_{Direct}^{N}$ and $1-\left( 1-P_{Direct}^1 \right)^N$ in Fig. \ref{fig:Cellular_outage_no_relay}. First, we observe that $P_{Direct}^{N}$ increases very fast with $N$, meaning that only a few non-cooperative eavesdroppers will block nearly all the secure connections from the BS to MU. Second, the outage probability is decreasing with $\alpha$ when $x$ is small while increasing when $x \approx 1$. Interestingly, it  suggests the MUs near the BS prefer severer path loss while the MUs near the cell edge prefer milder path loss. Finally, for the  cell edge MUs, i.e., $d_{sd}=R$, we have
\begin{eqnarray*}
  P_{Direct}^{1}(R)= \left \{ \begin{gathered}
  \begin{aligned}
    & \ln 2 \approx 0.693 \quad    & when \quad \alpha = 2 \hfill \\
    & \frac{2\pi}{3 \sqrt{3}} - \frac{2}{3} \ln2 \approx 0.747 \quad   &when \quad \alpha = 3 \hfill \\
    & \frac{\pi}{4} \approx 0.785 \quad  & when \quad \alpha = 4 \hfill \\
    \end{aligned}
  \end{gathered}  \right.
\end{eqnarray*}
It shows that the cell edge MUs have no secure connections to the BS with very high probability.

To deal with this issue, we then propose a heuristic relay placement strategy as follows.
The cell is first partitioned to $M$ sectors and MUs in every sector is served by a relay as depicted in Fig. \ref{fig:cellular_networks}. Obviously, the MUs located at both the cell edge and sector edge, like point $\textit{A}$ (see Fig. \ref{fig:cellular_networks}), have the largest outage probability. As $P_{Direct}^N(R)$ has the upper bound $1 - \left(1- P_{Direct}^{1}\left(R\right)\right)^N $, we consider only one eavesdropper and aim to minimize $P_{Direct}^1(R)$. We then search the optimal relay position and power to minimize the outage probability of such MUs by numerical simulation.
\begin{figure}
\begin{centering}
\includegraphics[scale=0.50]{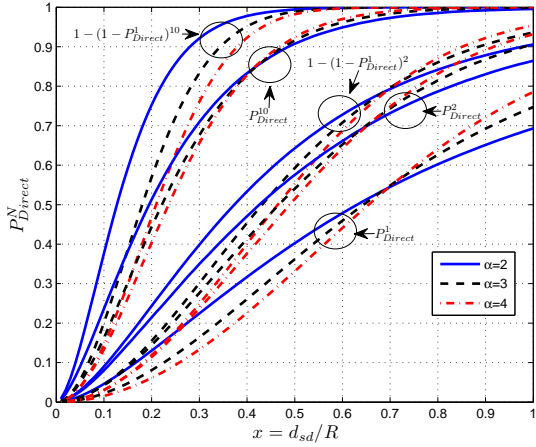}
\vspace{-0.1cm}
\caption{Outage probability vs the normalized distance between the BS and MU.} \label{fig:Cellular_outage_no_relay}
\end{centering}
\vspace{-0.3cm}
\end{figure}

\begin{figure}
\begin{centering}
\includegraphics[scale=0.58]{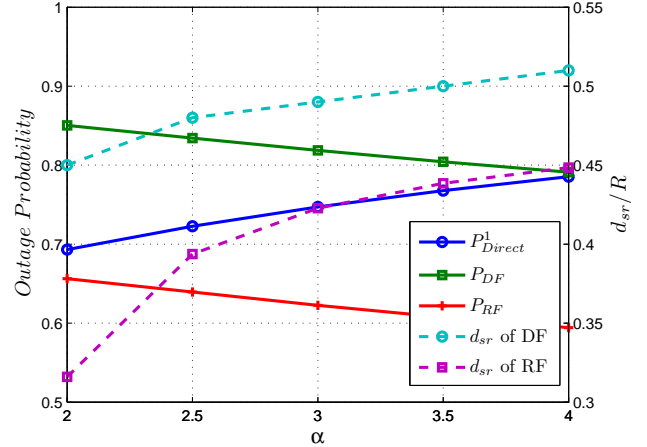}
\vspace{-0.1cm}
\caption{Outage probability of MU located at point A and the position of relay vs path loss exponent with one eavesdropper.} \label{fig:Cellular_outage_relay_position_vs_alpha}
\end{centering}
\vspace{-0.3cm}
\end{figure}

We show the numerical results in Fig. \ref{fig:Cellular_outage_relay_position_vs_alpha} where $M=6$, $N=1$, and each relay has the power constraint $p_r\leqslant p_s$. It is shown that in such cases, RF is better than direct transmission while DF is inferior to direct transmission. Moreover, the outage probability of direct transmission of the MUs at point $A$ is increasing with the path loss exponent $\alpha$, while the outage probability with DF or RF relay decrease with $\alpha$. Finally, the best relay position approaches to the cell edge as $\alpha$ increases.


\section{Conclusion}
\label{conclusion}
In this paper, we have considered relay placement for secure connection problem. Through analytical expressions and numerical results, we have shown that relay is beneficial for establishing secure connection for the four-node system and cellular networks. We also found that relay transmission is especially helpful when path loss is severer. Furthermore, it was shown that the RF relay strategy, by introducing different randomization in each hop, is much better than the traditional DF relay strategy.

\bibliographystyle{IEEEtran}
\bibliography{IEEEabrv,Tao_CL2012-0582}

\begin{thebibliography}{10}
\providecommand{\url}[1]{#1}
\csname url@samestyle\endcsname
\providecommand{\newblock}{\relax}
\providecommand{\bibinfo}[2]{#2}
\providecommand{\BIBentrySTDinterwordspacing}{\spaceskip=0pt\relax}
\providecommand{\BIBentryALTinterwordstretchfactor}{4}
\providecommand{\BIBentryALTinterwordspacing}{\spaceskip=\fontdimen2\font plus
\BIBentryALTinterwordstretchfactor\fontdimen3\font minus
  \fontdimen4\font\relax}
\providecommand{\BIBforeignlanguage}[2]{{%
\expandafter\ifx\csname l@#1\endcsname\relax
\typeout{** WARNING: IEEEtran.bst: No hyphenation pattern has been}%
\typeout{** loaded for the language `#1'. Using the pattern for}%
\typeout{** the default language instead.}%
\else
\language=\csname l@#1\endcsname
\fi
#2}}
\providecommand{\BIBdecl}{\relax}
\BIBdecl

\bibitem{Wyner_Bell75}
A.~D. Wyner, ``The wire-tap channel,'' \emph{Bell Syst. Tech. J.}, vol.~54,
  no.~8, pp. 1355--1367, Oct. 1975.

\bibitem{Lai_IT08}
L.~Lai and H.~El~Gamal, ``The relay--eavesdropper channel: Cooperation for
  secrecy,'' \emph{{IEEE} Trans. Inf. Theory}, vol.~54, no.~9, pp. 4005--4019,
  Sep. 2008.

\bibitem{Jeong_TSP11}
C.~Jeong and I.-M. Kim, ``Optimal power allocation for secure multicarrier
  relay systems,'' \emph{{IEEE} Trans. Signal Process.}, vol.~59, no.~11, pp.
  5428 --5442, Nov. 2011.

\bibitem{Dong_TSP10}
L.~Dong, Z.~Han, A.~Petropulu, and H.~Poor, ``Improving wireless physical layer
  security via cooperating relays,'' \emph{{IEEE} Trans. Signal Process.},
  vol.~58, no.~3, pp. 1875 --1888, Mar. 2010.

\bibitem{Krikidis_TWC09}
I.~Krikidis, J.~Thompson, and S.~Mclaughlin, ``Relay selection for secure
  cooperative networks with jamming,'' \emph{{IEEE} Trans. Wireless Commun.},
  vol.~8, no.~10, pp. 5003--5011, Oct. 2009.

\bibitem{Ng_TWC11}
D.~Ng, E.~Lo, and R.~Schober, ``Secure resource allocation and scheduling for
  {OFDMA} decode-and-forward relay networks,'' \emph{{IEEE} Trans. Wireless
  Commun.}, vol.~10, no.~10, pp. 3528 --3540, Oct. 2011.

\bibitem{Koyluoglu_IT12}
O.~O. Koyluoglu, C.~E. Koksal, and H.~El~Gamal, ``On secrecy capacity scaling
  in wireless networks,'' \emph{{IEEE} Trans. Inf. Theory}, 2012, to appear.

\bibitem{Goeckel_JSAC11}
D.~Goeckel, S.~Vasudevan, D.~Towsley, S.~Adams, Z.~Ding, and K.~Leung,
  ``Artificial noise generation from cooperative relays for everlasting secrecy
  in two-hop wireless networks,'' \emph{{IEEE} J. Sel. Areas Commun.}, vol.~29,
  no.~10, pp. 2067 --2076, Dec. 2011.

\bibitem{Popovski_TIFS09}
P.~Popovski and O.~Simeone, ``Wireless secrecy in cellular systems with
  infrastructure-aided cooperation,'' \emph{{IEEE} Trans. Inf. Forensics
  Security}, vol.~4, no.~2, pp. 242 --256, Jun. 2009.

\bibitem{Bloch_IT08}
M.~Bloch, J.~Barros, M.~Rodrigues, and S.~McLaughlin, ``Wireless
  information-theoretic security,'' \emph{{IEEE} Trans. Inf. Theory}, vol.~54,
  no.~6, pp. 2515 --2534, Jun. 2008.

\bibitem{Pinto_TIFS12a}
P.~Pinto, J.~Barros, and M.~Win, ``Secure communication in stochastic wireless
  networks -- {Part I}: Connectivity,'' \emph{{IEEE} Trans. Inf. Forensics
  Security}, vol.~7, no.~1, pp. 125 --138, Feb. 2012.

\bibitem{Zhou_TWC11}
X.~Zhou, R.~Ganti, and J.~Andrews, ``Secure wireless network connectivity with
  multi-antenna transmission,'' \emph{{IEEE} Trans. Wireless Commun.}, vol.~10,
  no.~2, pp. 425 --430, Feb. 2011.

\end{thebibliography}

\end{document}